\documentclass{article}

\usepackage{arxiv}

\usepackage[utf8]{inputenc} 
\usepackage[T1]{fontenc}    
\usepackage{hyperref}       
\usepackage{url}            
\usepackage{booktabs}       
\usepackage{amsfonts}       
\usepackage{nicefrac}       
\usepackage{microtype}      
\usepackage{lipsum}		
\usepackage{graphicx}
\usepackage{natbib}
\usepackage{doi}

\usepackage{amsmath, amsthm, amssymb}
\usepackage{graphicx}
\usepackage{subfigure}

 \newtheorem{thm}{Theorem}[section]
 \newtheorem{cor}[thm]{Corollary}
 \newtheorem{lem}[thm]{Lemma}
 \newtheorem{prop}[thm]{Proposition}
 \theoremstyle{definition}
 \newtheorem{defn}[thm]{Definition}
 \theoremstyle{remark}

 \numberwithin{equation}{section}

\newcommand{\m}{\mathcal}

\newcommand{\mf}{\mathfrak}
\newcommand{\st}{\hspace{0.10cm} | \hspace{0.10cm}}

\newcommand{\bdefi}{\begin{defn}}
\newcommand{\edefi}{\end{defn}}

\newcommand{\btw}{\begin{thm}}
\newcommand{\etw}{\end{thm}}

\newcommand{\blem}{\begin{lem}}
\newcommand{\elem}{\end{lem}}

\newcommand{\bcn}{\begin{con}}
\newcommand{\ecn}{\end{con}}

\newcommand{\bfc}{\begin{fc}}
\newcommand{\efc}{\end{fc}}

\newcommand{\n}{\noindent}

\newcommand{\beqq}{\begin{equation*}}
\newcommand{\eqq}{\end{equation*}}

\newcommand{\bc}{\begin{center}}
\newcommand{\ec}{\end{center}}

\newcommand{\bit}{\begin{itemize}}
\newcommand{\eit}{\end{itemize}}

\newcommand{\la}{\langle}
\newcommand{\ra}{\rangle}

\title{Three-Dimensional Affine Spatial Logics}

\author{ \href{https://orcid.org/0000-0003-4170-8665}{\includegraphics[scale=0.06]{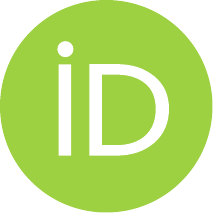}\hspace{1mm}Adam Trybus} \\
	Institute of Philosophy\\
	Jagiellonian University\\
	52 Grodzka St., Kraków 31-044, Poland\\
	\texttt{adam.trybus@uj.edu.pl} }

\renewcommand{\shorttitle}{Three-Dimensional Affine Spatial Logics}

\hypersetup{
pdftitle={Three-Dimensional Affine Spatial Logic},
pdfsubject={cs.LO},
pdfauthor={Adam Trybus},
pdfkeywords={spatial logic, affine geometry, qualitative spatial reasoning},
}

\begin{document}
\maketitle

\begin{abstract}
	We focus on a branch of region-based spatial logics dealing with affine geometry. The research on this topic is scarce: only a handful of papers investigate such systems, mostly in the case of the real plane. Our long-term goal is to analyse certain family of affine logics with inclusion and convexity as primitives interpreted over real spaces of increasing dimensionality. In this article we show that logics of different dimensionalities must have different theories, thus justifying further work on different dimensions. We then focus on the three-dimensional case, exploring the expressiveness of this logic and consequently showing that it is possible to construct formulas describing a three-dimensional coordinate frame. The final result, making use of the high expressive power of this logic, is that every region satisfies an affine complete formula, meaning that all regions satisfying it are affine equivalent.
\end{abstract}

\keywords{spatial logic \and affine geometry \and qualitative spatial reasoning}

\section{Introduction}

\paragraph{History and philosophy} Spatial logic, as the term has been used, can be viewed as built using any first-order language with geometrical interpretation, where variables range over geometrical entities and relation and function symbols are interpreted as geometrical relations and functions. However, in most instances, the name accompanies region-based, rather that point-based systems, meaning that the geometrical entities variables range over are not points but some collections of point. Although the name itself might be an invention of the early twenty-first century (see \cite{space}), spatial logics have rich and diverse background: after all, the notion of space is a staple in philosophy.\footnote{See \cite{casativarzi} for an excellent introduction into the intersection of philosophical and formal approaches to spatial reasoning.} One of the most famous treatments of space was proposed by Kant, who in \emph{Critique of Pure Reason} (\cite{kant}) argues that geometry is both synthetic and a priori (using his, now well-known, labels). Since the argument seem to be hinging on the existence of only one type of geometry, the development of new, non-Euclidean geometries in the nineteenth century was considered as a threat to Kantian views. Bertrand Russell, in one of his earliest publications, tries to defend Kant's approach (see \cite{Rus97}. Moreover, the approach of many nineteenth-century geometers was very philosophically informed. For example, Moritz Pasch --- most famous perhaps for figuring out the gaps in Euclid's reasoning --- viewed geometry as having a decidedly emprical basis (see e.g. \cite{pasch}). As a consquence, his analysis starts not with the Euclidean but rather with, what now is known as, affine geometry: one, where the notion of metric is not important. Thus, affine geometry can be viewed as emphasising \emph{qualitative} rather than \emph{quantitative} aspects of geometrical thinking. This theme is also important for contemporary researchers working within the so-called \emph{qualitative spatial reasoning} field. Logical formalisms that came from that field are sometimes called spatial logics, and those spatial logics that apply affine notions are the focus of our article. Although there has been some interest in such type of logic (see e.g. \cite{DavisGottsCohn:1999}, \cite{bennett} as well as \cite{Pratt:1999}, \cite{trybus}), it should be said that it pales in comparison with the research on topological spatial systems (see e.g. \cite{space} for a wide selection of topic related to topological formalisms). We believe that while there are good reasons for topological analysis, philosophical investigations provide additional justification for extending the work on affine systems. For example, Bertrand Russell, no doubt influenced by Pasch and others,\footnote{Notably by M. Pieri and F. Klein. For a more detailed description of their work see: \cite{pieri} and \cite{klein}, respectively.} took up the idea of the importance of non-numerical, qualitative, geometry and argued extensively for the primacy of projective and affine notions over the Euclidean ones (see \cite{Rus97}, \cite{Rus03} and \cite{TrybusRuss} for a discussion). Moreover, while Alfred~N. Whitehead's complex philosophical ideas influenced the development of region-based theories of space (see \cite{whitehead}), which are closely related to topology and mereology (a theory of part-whole relations, see \cite{simons}), he also devoted his attention to affine and projective geometry in \cite{Whi07} and \cite{Whi13} respecively. Finally, affine geometry can be said to be an intermediate geometry between the projective and Euclidean ones. Hence, it retains the status of non-numerical geometry and at the same time is less general than projective geometry and topology, thus remaining closer to our every-day experiences.

The region-based affine spatial logics --- the focus of our article --- are not the first attempts at logical analysis of this type of geometry. Alfred Tarski mentions affine geometry in his comparison between the developments in logic and geometry (see \cite{tarski:notions}, which is a written account of a talk he gave, which in turn relfects his ideas from before the war). Moreover, together with his student Lesław Szczerba, Tarski worked on point-based affine spatial logics (see \cite{tarskiszczerba}), which built on Tarski's earlier work on formalising Euclidean geometry (see \cite{tarskieuclid} and \cite{McFarland} for a detailed look at Tarski's involvement in geometry).\footnote{The article \cite{Nagel} is a fascinating summary of the influence of geometrical results on the development of logic.}

\paragraph{Constructing a spatial logic} If we were to custom-build a spatial logic, the first problem we are going to face is the choice of underlying geometric space. Many approaches have been studied, in most of them however either $\mathbb{R}^n$ for some $n$ or some more general topological space is considered. Having set on the underlying geometric space, say $X$, we are faced with another decision. Should the variables range over elements of $X$ or some subset $S \subseteq 2^X$? In the first case we would be talking about \emph{point-based} spatial logics, in the second about \emph{region-based} spatial logics, which is our focus here. As mentioned above, we place our work in the \emph{qualitative spatial reasoning} (QSR) subarea of symbolic AI. The adjective \emph{qualitative} in this context means that all the primitive relations and functions are of non-numerical nature. For example, consider a language with a single relation symbol $C$ understood as the contact relation. Intuitively two sets are in contact if their boundaries share at least one point. This spatial logic was investigated under many guises, most notably within the qualitative spatial reasoning paradigm. We are now faced with the following question: precisely what \emph{sort} of regions should we consider? We could obviously decide to consider all $S \subseteq 2^X$ for a given space $X$. Are there any reasons to consider a special class of regions rather than give them all an equal footing? One such reason is the admittedly vague notion of \emph{well-behavedness}. First of all, to smooth out the reasoning with regions, we would like to weed out as many ``special cases'' as possible. Assuming we are working with some topological space, this can be done by considering only \emph{regular} subsets of that space as plausible region-candidates. This gets rid of many a ``strange'' set e.g. of fractal nature. In the next step we need to decide whether we consider our regions to contain their boundaries or not. In the first case we end up with regular \emph{closed} sets and in the second case with regular \emph{open} sets. From a formal point of view, this is not an essential choice. In the remainder we will consider mainly regular open variants (and everything we say can be applied \emph{mutatis mutandis} to the regular closed case). The class of all regular open subsets of some topological space is already a good choice for the well-behaved regions.\footnote{This is by no means the final word in the quest for well-behavedness, see \cite{DanaScott}.} Apart from what has been mentioned already, by a well-known result the elements of the class of regular open subsets of some topological space form a Boolean Algebra, that is, operations of sum, product and complement of regular open sets conform to the laws of Boolean Algebra. We can do better still. We can look inside this class for some more refined region candidates. As is customary, we single out two classes: (regular open) \emph{polygons} and (regular open) \emph{rational polygons} (limiting ourselves to rational numbers). The fact that it is countable, makes the second subclass especially interesting from the point of view of computer science applications. The choice of geometric space and either point- or region-based approach dictates the choice of relations and functions that we are presented with. Within the qualitative spatial reasoning paradigm, non-numerical predicates on regions are considered, most notably contact and connectedness. Traditionally, spatial logics over languages containing relation and function symbols interpreted as relations and functions invariant under certain geometric transformations (topological, Eucidean etc.) are called accordingly as e.g. Euclidean, topological (spatial) logic. We follow this convention here. For example, consider an \emph{affine} spatial logic constructed in the following manner. Start with a language with two primitive symbols $\mf{conv}$ and $\leq$. Let them denote the following predicates defined on regular open rational polygonal subsets of $\mathbb{R}^2$. The symbol $\mf{conv}(a)$ is to be understood as ``region $a$ is convex'' and the symbol $a \leq b$ as ``region $a$ is a subset of region $b$''. It is an affine spatial logic, since convexity is an affine-invariant property. This spatial logic is in fact one that we are concerned the most within this article. The last choice made in constructing a spatial logic concerns the syntactical complexity of the language we want to use. In our case, we work with standard first-order logic.

\paragraph{The focus of this article} The order of the article is as follows. After some more technical remarks regarding region-based theories of space and affine geometry, finally definining the structures that are important for us. Then, we describe in short the most important results obtained in \cite{trybus}. This is done partially to introduce certain (visual) intuitions that are easier to grasp in the two dimensional case but that carry over, to some extent, to the three-dimensional case. Next, we describe some more general results regarding the family of structures that we have defined: namely that they all have different theories. Finally, we fix our attention on the three-dimensional extension of the two-dimensional logic analysed in \cite{trybus}. We prove a number of expressiveness results that are helpful in establishing a result similar to one of the main theorems of \cite{Pratt:1999}, regarding the existence of formulas that are satisfied only by affine-equivalent regions.

\section{General setup}

Let $\m{L}_{\mf{conv},\leq}$  be a first-order language with two predicates: binary $\leq$ and unary $\mf{conv}$. We work with an $\m{L}_{\mf{conv},\leq}$-structure with variables ranging over the set of regular open rational polygons of the real plane; $\leq$ interpreted as the inclusion relation and $\mf{conv}$ as a property of being convex. We start with defining a notion of a \emph{regular open set}.

\bdefi Let $S$ be a subset of some topological space. We denote the \emph{interior} of $S$ by $S^0$ and the \emph{closure} of $S$ by $S^-$. $S$ is called \emph{regular open} if $S = (S)^{-^0}$.\edefi

The following result is standard.

\begin{prop} The set of regular open sets in $X$ forms a Boolean algebra $RO(X)$ with top and bottom defined by $1 = X$ and $0 = \emptyset$, and Boolean operations defined by $a\cdot b = a \cap b$, $a + b = (a \cup b)^{-0}$ and $-a = (X \setminus a)^{0}$.\end{prop}

We restrict our attention to certain well-behaved regular open sets (see our remarks in the introduction). Let us start with a staple topological space used in \emph{QSR}, namely $\mathbb{R}^2$. Note that every line in $\mathbb{R}^2$ divides $\mathbb{R}^2$ into two domains, called \emph{half-planes}. Open half-planes are regular open sets, hence we can speak about the sums, products and complements of such half-planes in $RO(\mathbb{R}^2)$. By a \emph{regular open rational polygon} we mean a Boolean combination in $RO(\mathbb{R}^2)$ of finitely many half-planes bounded by lines with rational coefficients in $\mathbb{R}^2$. We denote the set of all regular open rational polygons in $\mathbb{R}^2$ by $ROQ(\mathbb{R}^2)$. Note that $ROQ(\mathbb{R}^2)$ is a Boolean subalgebra of $RO(\mathbb{R}^2)$. The notion of regular open rational polygon can be easily extended to that of a \emph{polytope}, when considering dimensions greater than $2$. In general, we write $ROQ(\mathbb{R}^n), n \in \mathbb{N}$, to denote the set of all regular open rational polytopes of dimension $n$ (all the mentioned results carry over from the two-dimensional case).

\bdefi A set $S \in \mathbb{R}^n$ is called \emph{convex} if for all $\lambda_1,\lambda_2 \in \mathbb{R}$, such that $\lambda_1, \lambda_2 \geq 0$ and $\lambda_1 + \lambda_2 = 1$ and for all $x \in S$, $$\lambda_{1} x + \lambda_{2}y \in S.$$\edefi

Finally, let us introduce the family of $n$-dimensional structures we will be interested in.

\bdefi Let $\mf{M}^n = \la ROQ(\mathbb{R}^n), \mf{conv}^{\mf{M}}, \leq^{\mf{M}} \ra$, where\\

$ \leq^{\mf{M}} = \{ \langle a,b \rangle \in ROQ(\mathbb{R}^n) \times ROQ(\mathbb{R}^n) \st a \subseteq
b\}$; \\

$ \mf{conv}^{\mf{M}} = \{ a \in ROQ(\mathbb{R}^n) \st a \mbox{ is convex} \}$.



\edefi

\n We sometimes refer to $\mf{M}^n $ as a rational model (of dimension $n$) and often drop the associated superscripts and subscripts if it does not lead to confusion.\\

\n In our exposition we follow the standard notational conventions. In particular, if $\phi$ is a formula, $\phi(x_1, \ldots, x_n)$ means that $\phi$ has at most $n$ variables: $x_1 , \ldots, x_n$. Also, if an $n$-tuple of regions $a_1, \ldots, a_n$ satisfy $\phi$ in $\mf{M}$, we write $\mf{M}\models\phi[a_1 , \ldots, a_n]$. However, we eschew formal clutter and whenever possible avoid dissecting the text with lemmas in favour of verbal description of results (especially the simpler ones) preserving the flow of thought. It should be noted, however, that all results can be easily converted into a more formalised description.\\

We also need a few simple facts regarding affine geometry. The following generalises the notion of an affine transformation in $\mathbb{R}^2$ to any dimension $n$. Recall that an $n \times n$ matrix $\mathcal{A}$ is invertible if there exists a $n \times n$ matrix $\mathcal{B}$ with $\mathcal{AB} = \mathcal{I}$, where $\mathcal{I}$ is the identity matrix; $\mathcal{A}$ is orthogonal if $\mathcal{AA}^T = \mathcal{I}$, where $\mathcal{A}^T$ is the transpose of $\mathcal{A}$.\\

\bdefi

An \emph{(n-dimensional) affine transformation} of $\mathbb{R}^n$ is a function $\tau: \mathbb{R}^n \to \mathbb{R}^n$ of the form $$ \tau(x) = \mathcal{A}x + b,$$ where $\mathcal{A}$ is an invertible $n \times n$ matrix and $b \in \mathbb{R}^n$.

\edefi

Note that affine transformations map straight lines to straight lines, preserve parallelism and ratios of lengths along parallel straight lines.\footnote{Hence, the properties of being a straight line, of lines being parallel and of being a ratio of a certain type are all affine-invariant.} Moreover, it is a standard result that the set of affine transformations forms a group under the operation of composition of functions. We say that two regions are \emph{affine-equivalent} if there is an affine transformation from one region to another (this notion naturally extends to sequences of regions).

\section{Two dimensions} The papers \cite{DavisGottsCohn:1999}, \cite{Pratt:1999} together with \cite{trybus} deal with various systems related to $\mf{M}^2$. The two-di\-men\-sio\-nal rational model turns out to be very expressive. Firstly note that the Boolean operations are clearly $\m{L}_{\mf{conv},\leq}$-definable (as are $0$ and $1$; for details see below, Theorem \ref{thm:boolean}). The paper \cite{Pratt:1999} showed that a number of interesting properties are definable in $\mf{M}^2$. It is easy to see that we can define a formula satisfied in the two-dimensional rational model if and only if the respective region is a half-plane (half-plane is the only region such that both it and its complement are convex). For the remainder of this paragraph, we use letters $l, m, n$ etc. (possibly with subscripts) to denote such half-planes but sometimes we abuse the convention and use the same symbols to denote the lines bounding the half-planes in question. With that in mind, \cite{Pratt:1999} showed that there is a formula involving two variables satisfiable in the two-dimensional rational model if and only if the two regions are half-planes with \emph{coincident} bounding lines. Similarly, there is a formula, such that the two regions involved are half-planes with \emph{parallel} bounding lines. Note that in affine geometry a \emph{coordinate frame} is defined as follows.

\bdefi\label{def:crf}

Let $l, m, n$ be any non-parallel, non-coincident lines with $l \cap m = \mathbf{O}$, $l \cap n = \mathbf{I}$ and $m \cap n = \mathbf{J}$. We say that $l, m, n$ form a \emph{coordinate frame}.

\edefi

Fig. \ref{fig:1} provides some examples of coordinate frames. Since the construction involves all the notions expressible in the two-dimensional model, we can ``talk'' about coordinate frames within that spatial logic.

\begin{figure}[h]
 \begin{center}
  
    \subfigure[]{\includegraphics[scale=0.32]{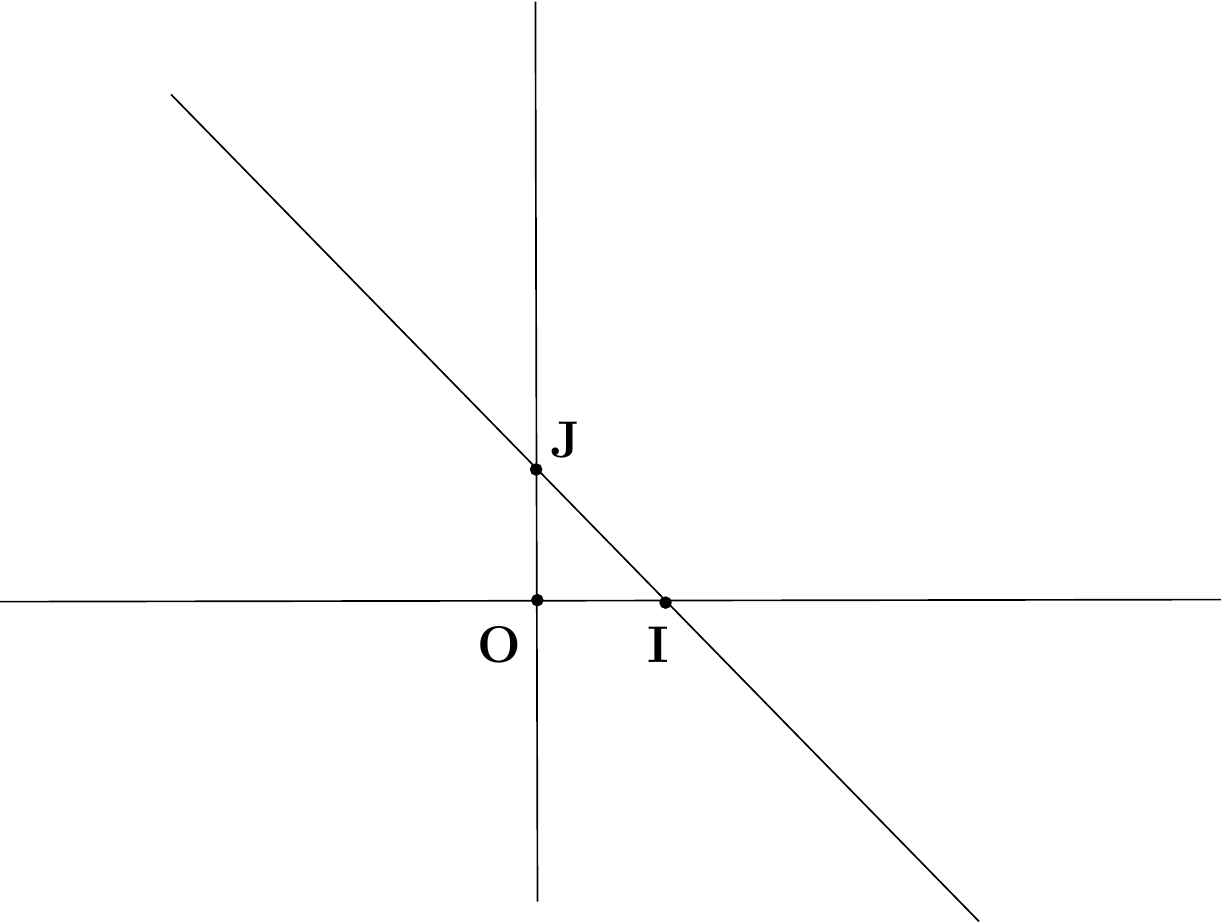}}
    \subfigure[]{\includegraphics[scale=0.33]{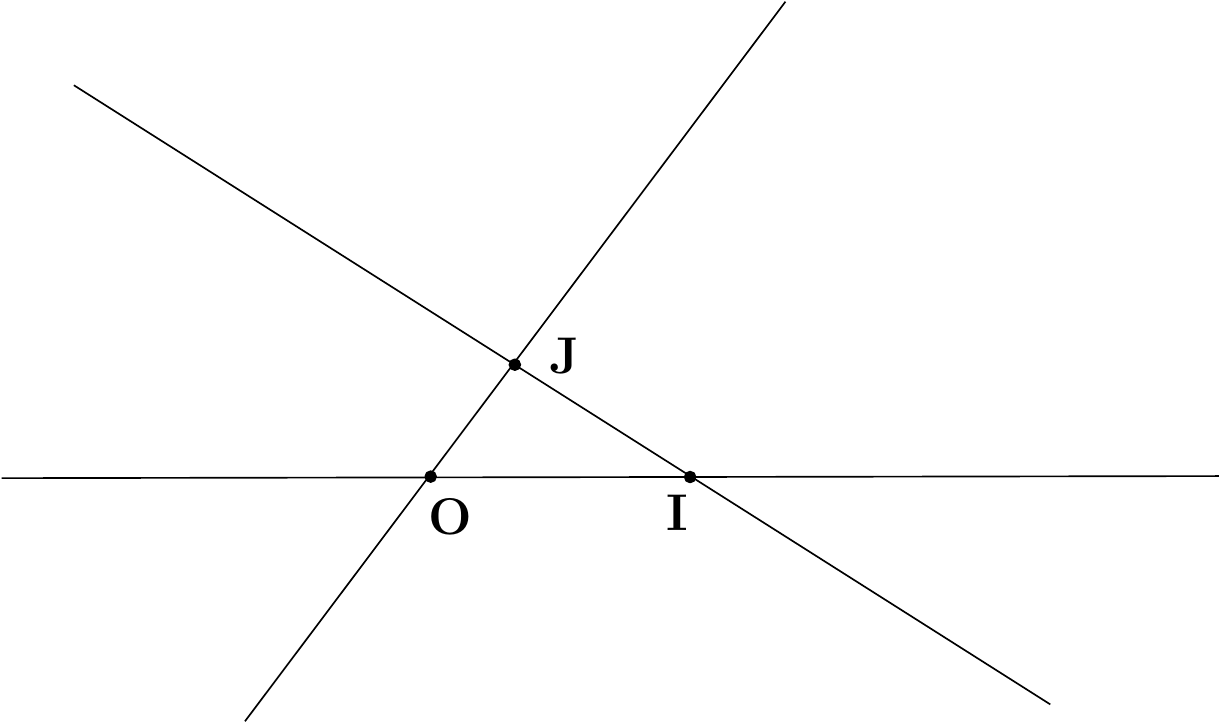}}
  
 \end{center}

 \caption{Example coordinate frames ($m$ is the horizontal and $l$ the vertical line in (a), the image in (b) can be thought of as an affine transformation of that of (a)). }

 \label{fig:1}
 \end{figure}

Now, the papers \cite{Pratt:1999} and \cite{trybus} show, in a sequence of results, that there exist formulas that allow fixing any rational half-plane with respect to a given coordinate frame. This is done by further exploring the expressivity of the model. Note that \cite{DavisGottsCohn:1999} shows that if two regions are affine-equivalent, then for certain affine spatial logics these satisfy the same formulas. An analogous theorem, relating the language $\m{L}_{\mf{conv},\leq}$  is proved in \cite{Pratt:1999}. Using the fixing formulas, the converse theorem is shown to hold in the case of $\mf{M}_2$.

\btw[\cite{Pratt:1999}]\label{thm:ian}
Every $n$\nolinebreak-tuple in $\mf{M}^2$ satisfies an $\m{L}_{\mf{conv},\leq}$-formula $\phi$ with the following property: any two $n$-tuples satisfying $\phi$ are affine-e\-qui\-va\-lent.
\etw

As indicated above, the proof relies on constructing certain formulas that allow us to ``talk'' about rational polygons and fixing their bounding lines in a certain manner, heavily relying on the expressivity results outlined above. (The paper \cite{trybus} provides details for this construction.) Moreover, the paper \cite{trybus} uses these ``fixing'' formulas to provide an axiom system for the two-dimensional model, which is proved to be sound and complete. The axioms express a number of properties e.g. that there are at least three regions such that lines bounding them form a coordinate frame or that if a region is a Boolean combination of half-planes, then it is convex if and only if it is a product of some of these half-planes. However, for the most part, the axioms secure certain properties of these fixing formulas. The axiom system is also equipped with two infinitary rules of inference stating that every half-plane can be fixed in reference to a given coordinate frame and that every region is a Boolean combination of some half-planes. Let us finally note that our main result in the present article closely mimicks that described in Theorem \ref{thm:ian}.

\section{Beyond two dimensions}

What can be known about the rational models of dimensions greater than two? Even at this stage, one can indeed make some statements about the relations among such models. Recall the well-known Helly's theorem.

\btw[Helly]

Let $A$ be a finite class of $N$ convex sets in $\mathbb{R}^n$ such that $N \geq n+1$ and each $n+1$-element subclass of $A$ has a non-empty
intersection. Then all $N$ elements of $A$ have a non-empty intersection.

\etw

\begin{figure}[h]

\centering
\includegraphics[scale=0.4]{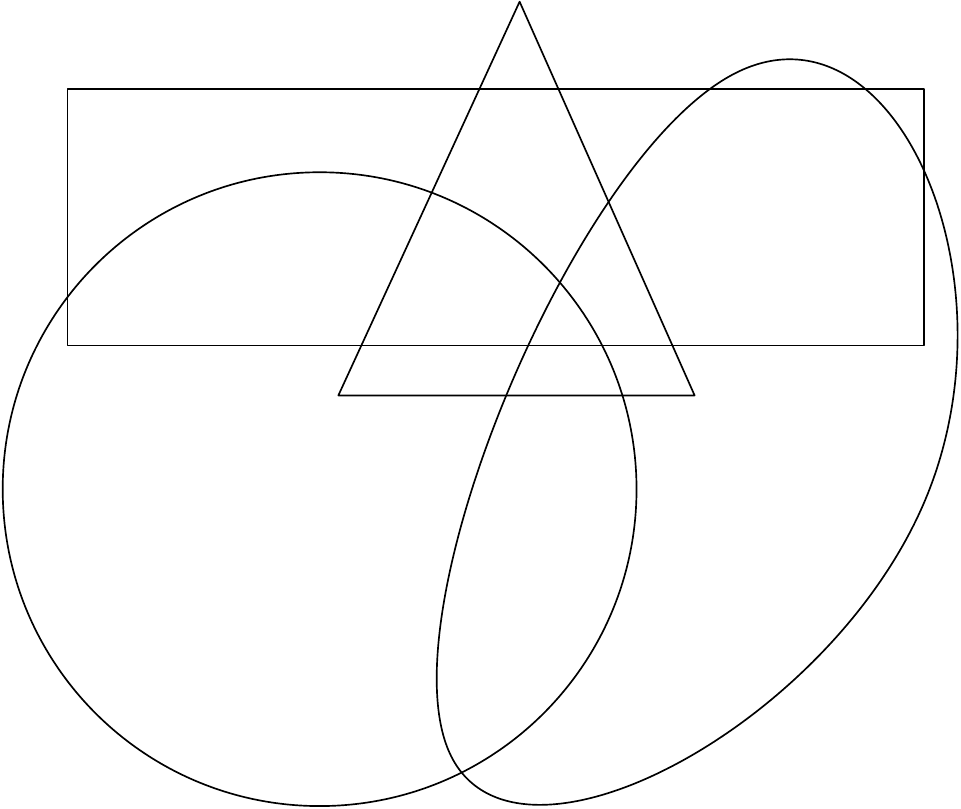}\caption[Helly's Theorem.]{A
very simple example of Helly's Theorem in $\mathbb{R}^2$.}

\end{figure}

First off, note that for all $n$, we have the following easy result.

\btw\label{thm:boolean}

Let $\mf{M}^n$ be a rational polygonal model. Then, the Boolean operators: product ($\cdot$), sum ($+$) and complement ($-$) are definable in $\mf{M}^n$. Moreover the top (1) and bottom (0) elements are also definable.

\etw

\begin{proof} For simplicity, we shall represent such formulas in their (infix) form as $x \cdot y$, $x + y$ and $-x$ and top and bottom as $1$ and $0$, instead of more correct but cumbersome standard notation as formulas in our language (which is what they really are), abusing the symbol of equality to also render situations like $x\cdot y = 0$. Now for the definitions of formulas. The following formula is satisfiable in $\mf{M}^n$ if and only if the region represented by the variable $m$ is the product of the regions represented by $x$ and $y$ respectively: $(m \leq x) \wedge (m \leq y) \wedge \forall w (w \leq x \wedge w \leq y \to w \leq m)$. An analogous formula can be constructed for the sum (all easily expressible in our language). The top can be defined as $\forall y (y \leq x)$ and the bottom as $\forall y (y = x \vee \lnot (y \leq x))$. Moreover, the following formula is satisfiable in $\mf{M}^n$ if and only if the region represented by $y$ is the complement of the region represented by $x$: $y\cdot x = 0 \wedge \forall w (w \cdot x = 0 \to w \leq y)$.\end{proof}

Now, consider the following formula $\phi(x_1,\ldots,x_N) :=$  $$\displaystyle\bigwedge_{I \subseteq 2^S} \prod_{i \in I} x_i \neq 0 \wedge \displaystyle \bigwedge_{1 \leq j \leq N} \mf{conv}(x_j) \to \displaystyle \prod_{1 \leq j \leq N} x_j \neq 0,$$

\n where $S = \{1, \ldots, N \}$.\\

\n This formula\footnote{We use $\prod$ and $\sum$ as abbreviations for finite products and sums, respectively. We also use $\pm a$ to denote region $a$ or its complement.} ``says'' in any $n$-dimensional model that regions $r_1, \ldots, r_N$ have non-empty intersection if each $r_j$ is convex and for every subset of $\{ r_1, \ldots, r_N\}$, its elements have a non-empty intersection.

\btw

For a given $n$ there exists a set of formulas $\Phi_n$ expressing the Helly's theorem in $\mf{M}^n$.

\etw

To see that this is the case, consider $\phi_{N} := \forall x_1 \ldots \forall x_N \phi(x_1,\ldots,x_N)$. We define $\Phi_n = \{ \phi_{N} \st  N \geq n+1 \mbox{ and } |I| = n+1\}$.\\

Recall that the theory of a structure is the set of all sentences valid in that structure.

\btw

The theory of $\mf{M}^n \neq \mf{M}^{n+1}$ for all $n$.

\etw

To see that this is the case, observe that for some $\phi_{N} \in \Phi_n$ we have $\mf{M}^n \models \phi_{N}$ but $\mf{M}^{n+1} \not \models \phi_{N}$. \\

Therefore, we can say that these models are indeed \emph{different}. However, if one were to extend the axiomatisation results from \cite{trybus} to dimensions greater than two --- and this is indeed our long-term goal --- the models should be also shown to be \emph{similar} in some other respect. Namely, the first task would be to see whether the notion of a coordinate frame can be expressed in such models in general.

\section{Three dimensions}

\paragraph{Basic expressivity} First of all, notice that the formula $\mf{conv}(x) \wedge \mf{conv}(-x)$ is satisfiable in $\mf{M}^3$ only by regions that are half-spaces. Since the plane bounding such half-spaces is unique, this also allow us to talk indirectly about such planes. It is convenient to be able to talk about a number of different half-spaces (planes); hence we introduce the following abbreviation. $$\mf{hs}_{n}(x_1,\ldots,x_n) := \displaystyle\bigwedge_{1 \leq i \leq n} \mf{conv}(x_i) \wedge \mf{conv}(-x_i) \wedge \displaystyle\bigwedge_{\substack{1 \leq i \leq n, \\ 1 \leq j \leq n, \\
i \neq j}} x_i \neq x_j \wedge x_i \neq -x_j$$

\n Next, we see that we can talk about parallel planes by means of the following formula. $$ \mf{hs}_{2}(x,y) \wedge ( (x \cdot y = 0 \vee x \cdot -y = 0) \vee (-x \cdot y = 0 \vee -x \cdot -y = 0) )$$

\n This formula ``says'' in $\mf{M}^3$ that the two regions are distinct half-spaces (with distinct bounding planes) and that it is either that the first region has a non-empty intersection with the other or that its complement has this property. (Note that the main disjunction in the brackets operates really as an exclusive ``or'' when the two half-spaces are different.) Similarly, the following formula expresses the fact that two planes meet in the single line. Since such lines are unique, we can also --- albeit indirectly --- talk about lines in $\mf{M}^3$. $$\mf{line}(x,y) := \mf{hs}_{2}(x,y) \wedge \lnot ( (x \cdot y = 0 \vee x \cdot -y = 0) \vee (-x \cdot y = 0 \vee -x \cdot -y = 0) )$$

Finally, consider the case when the following is satisfied: $$\mf{line}(y_1,y_2) \wedge \mf{line}(y_1,y_3) \wedge \mf{line}(y_2,y_3)$$ (we thus assume this piece of formalism to be a part of all formulas described in the remainder of this paragraph). This formula simply says that all the planes bounding the three half-space have a non-empty intersection with each other. Consider the following three configurations:

\begin{description}

\item[(i)] \textbf{a fan:} where all the planes meet in a single line;


  
\item[(ii)] \textbf{a prism:} where two of the planes meet in a line not on the third plane and meet the third plane in two separate, parallel lines;


\item[(iii)] \textbf{a corner:} where two of the planes meet the third plane in two separate, non-parallel lines and meet each other in a line that passes through the third plane.



\end{description}

Since in all the above cases, the number of domains into which the entire space is being partitioned changes (6 domains for a fan, 7 for a prism and 8 for a corner) and it can be expressed in terms of products of respective half-spaces or their complements, one can build formulas describing all three cases in $\mf{M}^3$. Noting that there are $8$ non-empty intersections possible in total, in the case of a corner, one enforces a non-empty intersection of all the half-spaces by adding the formula $$\lnot \exists x \lnot \exists y \lnot \exists z (((x = y_1 \vee x = -y_1) \wedge (y = y_2 \vee y = -y_2) \wedge (z = y_3 \vee z = -y_3))$$ $$\wedge (x \cdot y \cdot z = 0))$$

Directly, and assuming that $\mf{line}(y_1,y_2) \wedge \mf{line}(y_1,y_3) \wedge \mf{line}(y_2,y_3)$ is satisfied, this formula ``says'' that any three (out of six in total --- remember we alway have a half-space and its complement) half-spaces bounded by some planes have a non-empty intersection. 

Next, in the case of a prism, one simply adds the formula
$$\exists x \exists y \exists z (((x = y_1 \vee x = -y_1) \wedge (y = y_2 \vee y = -y_2) \wedge (z = y_3 \vee z = -y_3))$$ $$ \wedge (x \cdot y \cdot z = 0))$$

\noindent forcing the existence of a non-empty intersection. However, when paired with
$$\lnot \exists x' \lnot \exists y' \lnot \exists z' (((x' = y_1 \vee x' = -y_1) \wedge (y' = y_2 \vee y' = -y_2) \wedge (z' = y_3 \vee z' = -y_3))$$ $$\wedge (x \neq x' \vee y \neq y' \vee z \neq z') \wedge (x' \cdot y' \cdot z' = 0)).$$

\n Thus, the end effect is \emph{only} one empty intersection. Finally, in the case of a fan, one has to force precisely two non-empty intersections. This is done by stringing together the following:

$$\exists x \exists y \exists z (((x = y_1 \vee x = -y_1) \wedge (y = y_2 \vee y = -y_2) \wedge (z = y_3 \vee z = -y_3))$$ $$ \wedge (x \cdot y \cdot z = 0))$$

and

$$\exists x' \exists y' \exists z' (((x' = y_1 \vee x' = -y_1) \wedge (y' = y_2 \vee y' = -y_2) \wedge (z' = y_3 \vee z' = -y_3))$$ $$ \wedge (x' \cdot y' \cdot z' = 0))$$

with the condition that
$$(x \neq x' \vee y \neq y' \vee z \neq z')$$

together with
$$\lnot \exists x'' \lnot \exists y'' \lnot \exists z''$$ $$(((x'' = y_1 \vee x'' = -y_1) \wedge (y'' = y_2 \vee y'' = -y_2) \wedge (z'' = y_3 \vee z'' = -y_3))$$ $$\wedge ((x \neq x'' \vee y \neq y'' \vee z \neq z'') \vee (x' \neq x'' \vee y' \neq y'' \vee z' \neq z'')) \wedge (x'' \cdot y'' \cdot z'' = 0))$$

The above are admittedly long-winded but relatively simple and repetitive formulas. We hide the details under the self-explanatory abbreviations $\mf{fan}(x,y,z)$, $\mf{prism}(x,y,z)$ and $\mf{corner}(x,y,z)$.\\

Having established this, let us note that the case (iii) provides a basis for a coordinate frame. For the remainder of this section, we focus on fleshing out one of the ways of defining a coordinate frame in $\mf{M}^3$. Consider the formula $\mf{frame}(y_1,y_2,y_3,y') := $
$$\mf{corner}(y_1,y_2,y_3) \wedge \mf{line}(y_1,y') \wedge \mf{line}(y_2,y') \wedge \mf{line}(y_3,y')$$
It is satisfiable by a tuple of elements $a_1, a_2, a_3, a'$ only when these are half-spaces such that the planes bounding the first three of them form a corner and the plane bounding the fourth half-space form a prism with each pair of these planes. The three lines that lie at the pairwise intersections of the planes $a_1, a_2, a_3$ will determine the \emph{axes} of the coordinate frame. Next, $a'$ meets the remaining planes at three distinct lines that intersect pairwise on each of the axes: the points of intersection of each pair of such lines and an axis will be marked with a point, called the \emph{unit of measurement} (akin to the points \textbf{I} and \textbf{J} in Figure \ref{fig:1} but for all the three planes involved). Thus, there are three axes and three units of measurement.

\paragraph{Addition and multiplication}

Consider two planes intersecting a third one in two lines. These lines are coincident, if the planes themselves are. Let us assume that $\mf{line}(y_1,y) \wedge \mf{line}(y_2,y)$ is satisfied. By adding
$$y_1 = y_2 \vee y_1 = -y_2$$

we define a relevant formula, denoted $\mf{coincident}_2(y_1,y_2,y)$.\footnote{$\mf{coincident}_2$ should be understood as defining coincidence in two dimensions. Similarly for other notions used in this paragraph.} Similarly, such lines are parallel, if the planes are. Therefore if we add $$\lnot (y_1 = y_2 \vee y_1 = -y_2)$$ $$\wedge (y_1 \cdot y_2 = 0 \vee y_1 \cdot -y_2 = 0 \vee -y_1 \cdot y_2 = 0 \vee -y_1 \cdot -y_2 = 0)$$

\n we obtain a formula (denoted $\mf{parallel}_2(y_1,y_2,y)$) satisfied in $\mf{M}^3$ if and only if the relation of parallelism holds between the respective lines. Also, when lines in a plane are not coincident or parallel, they have to meet in a single point. Thus, we can add the following constraints
$$\lnot \mf{coincident}_2(y_1,y_2,y) \wedge \lnot \mf{parallel}_2(y_1,y_2,y),$$

\n defining a formula $\mf{point}_2(y_1,y_2,y)$ with the obvious interpretation. We need these expressivity results to define important operations on line segments found on the planes forming the coordinate frame. We start with defining addition in a plane (following \cite{Bennett:1995}):

\bdefi
We say that $\overline{\mathbf{OC}}$ is the result of the \emph{addition} of $\overline{\mathbf{OA}}$ and $\overline{\mathbf{OB}}$ and write  
$ \overline{\mathbf{OA}} + \overline{\mathbf{OB}} = \overline{\mathbf{OC}}$ if and only if the following lines can be found (see Fig. \ref{fig:add}):

\bit

\item [ (a) ] $l_1$, $l_3$ meeting at a point $\mathbf{O}$;

\item [ (b) ] $m$ parallel to $l_3$;

\item [ (c) ] $l_A$, meeting $l_3$ at a point $\mathbf{A}$ and parallel or coincident with $l_1$;

\item [ (d) ] $l_B$, meeting $l_3$ at a point $\mathbf{B}$ and such that $l_B, l_1, m$ meet at a single point $\mathbf{J}$;

\item [ (e) ] $l_C$, meeting $l_3$ at a point $\mathbf{C}$ and parallel or coincident with $l_B$ and such that $l_A, l_C$ and $m$ meet at a single point $\mathbf{M}$.

\eit 

\edefi

\begin{figure}[h]
\centering
\includegraphics[scale=0.55]{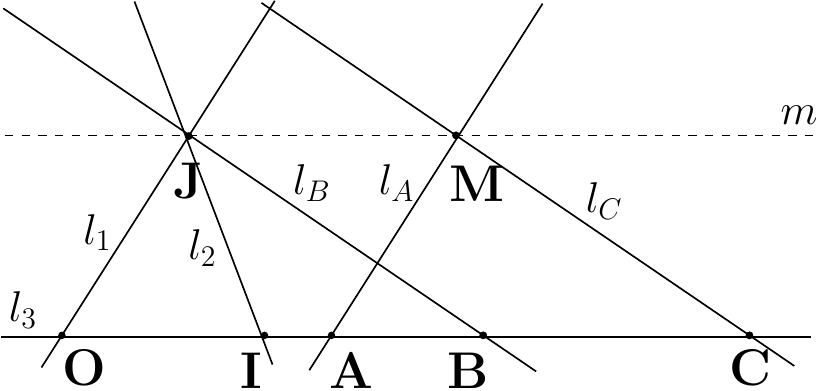}\caption{$ \overline{\mathbf{OA}} + \overline{\mathbf{OB}} = \overline{\mathbf{OC}}$.}\label{fig:add}

\end{figure}

Say, for simplicity, that we adopted the same notational conventions for objects in the structure $\mf{M}^3$ (remembering that these can refer to the relevant lines only indirectly and that directly these denote objects in the three-dimensional space!). Furthermore, let $$\mf{M}^3 \models \mf{frame}[l_1,l_3,n,n'],$$ for some $n, n'$ (with the plane bounding $n'$ meeting the plane bounding $n$ in the line $l_2$, thus defining the unit of measurement): this takes care of (a).\footnote{Note that our entire two-dimensional construction is `happening' in the plane bounding $n$.} We can enforce (b) by $$\mf{M}^3 \models \mf{parallel}[m,l_3,n];$$ (c) with $$\mf{M}^3 \models \mf{point}_2[l_A,l_3,n] \wedge (\mf{parallel}_2[l_A,l_1,n] \vee \mf{coincident}_2[l_A,l_1,n]);$$ (d) with $$\mf{M}^3 \models \mf{point}_2[l_B,l_3,n] \wedge \mf{corner}_2[l_B,l_1,m];$$ and (e) with $$\mf{M}^3 \models \mf{point}_2[l_C,l_3,n] \wedge (\mf{parallel}_2[l_C,l_B,n] \vee \mf{coincident}_2[l_C,l_B,n]).$$ Thus, we can construct a formula $\mf{add}_2(y_1,y_3,y_A,y_B,y_C,y,y',z)$ such that $\mf{M}^3 \models \mf{add}_2[l_1,l_3,l_A,l_B,l_C,n,n',m]$ if and only if $ \overline{\mathbf{OA}} + \overline{\mathbf{OB}} = \overline{\mathbf{OC}}$ (assuming the naming conventions above).

Similarly, again after \cite{Bennett:1995}, let us define multiplication in a plane.

\bdefi
We say that $\overline{\mathbf{OC}}$ is the result of \emph{multiplication} of $\overline{\mathbf{OA}}$ and $\overline{\mathbf{OB}}$ and write $ \overline{\mathbf{OA}} \cdot \overline{\mathbf{OB}} = \overline{\mathbf{OC}}$ if and only if the following lines can be found (see Fig. \ref{fig:mult}):

\bit

\item [ (a) ] $l_1$, $l_3$ meeting at a point $\mathbf{O}$ and $l_2$ meeting $l_1$ at a point $\mathbf{J}$ and $l_3$ at a point $\mathbf{I}$;

\item [ (b) ] $l_A$ meeting $l_3$ at a point $\mathbf{A}$ and parallel or coincident with $l_2$;

\item [ (c) ] $l_B$ meeting $l_3$ at a point $\mathbf{B}$ and such that $l_B, l_1, l_2$ meet at a single point ($\mathbf{J}$);

\item [ (d) ] $l_C$ meeting $l_3$ at a point $\mathbf{C}$, parallel or coincident with $l_B$ and such that $l_C, l_A, l_1$ meet at a single point $\mathbf{M}$.

\eit

\edefi

\begin{figure}[h]
\centering
\includegraphics[scale=0.65]{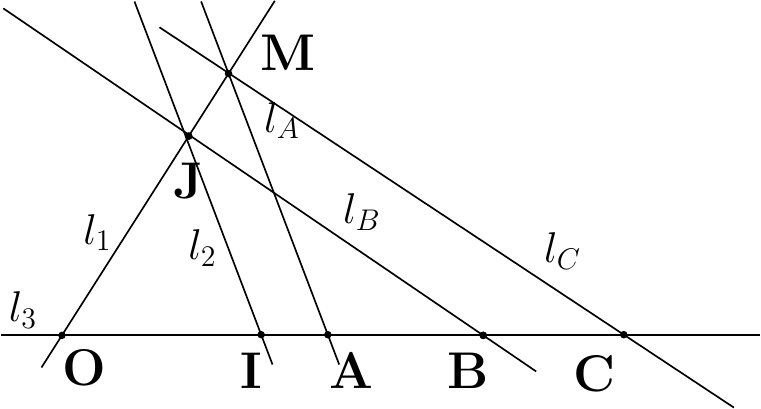}\caption{$ \overline{\mathbf{OA}} \cdot \overline{\mathbf{OB}} = \overline{\mathbf{OC}}$.}\label{fig:mult}

\end{figure}

\n Given the similarity of the constraints for multiplication to those for addition, it should be clear now that there is a formula $\mf{multiply}_2(y_1,y_3,y_A,y_B,y_C,y,y')$ such that $\mf{M}^3 \models \mf{multiply}_2[l_1,l_3,l_A,l_B,l_C,n,n']$ if and only if $ \overline{\mathbf{OA}} \cdot \overline{\mathbf{OB}} = \overline{\mathbf{OC}}$ (assuming the naming conventions above).

\begin{cor}

Addition and multiplication are definable in every plane bounding the half-spaces used in defining a coordinate frame in $\mf{M}^3$.  

\end{cor}

To obtain this simple consequence one needs to change what counts as the plane of reference (where these operations are defined by means of the above-described formulas).

\paragraph{Affine completeness} In this part, we show how to obtain results analogous to those presented in \cite{Pratt:1999} with regard to $\mf{M}^2$. Assume for now that we work in a specified plane of reference with the coordinate frame defined as above. How one would go about actually defining the numbers on the $x$-axis (the $y$-axis being analogous)? Well, we can start by defining these in terms of distance. So $0$ would be $\overline{\mathbf{OO}}$ and $1$ would be $\overline{\mathbf{OI}}$ with the rest of the natural numbers obtained by ``repeating'' the construction of $\overline{\mathbf{OI}}$. This is in fact how things are done in \cite{Pratt:1999} and \cite{trybus}, so the interested reader is encouraged to consult these sources. In this article, however, we propose a slightly different solution, using the fact that addition is expressible to define a \emph{successor} formula instead. Say, we defined $0$ as the point of intersection of the two axes, and $1$ as the point of intersection of $l_3$ (assuming previous conventions) with the $x$-axis represented by $l_1$. The successor formula can be defined so that, starting with $\overline{\mathbf{OO}}$ as the base case (this is done by simply constructing a formula satisfiable only by those regions whose associated lines determine the same point as the lines determining the axes), one keeps on adding $\overline{\mathbf{OI}}$. More formally (yet still, with a lot of ancillary details left out for readability), assume that a line $m$ crosses the line $l_1$ at a point $\mathbf{M}$, such that $\overline{\mathbf{OM}} = n\overline{\mathbf{OI}}, n \in \mathbb{N}$. We then define the successor formula that involves two important elements: the regions representing $m$ and another line $m'$, such that $m'$ is the result of adding $\overline{\mathbf{OM}}$ and $\overline{\mathbf{OI}}$ (clearly doable, by the above).\footnote{It should be obvious at this point that the theory of $\mf{M}^3$ is undecidable.} This takes care of the natural numbers on the line. The following result shows that we can extend this to any rational number.

\btw\label{thm:8}

Assuming the coordinate frame setup above and all the introduced shorthands, let $m$ be a line crossing the axis at a point $\mathbf{M}$. Then there exists a formula satisfiable in $\mf{M}^3$ if and only if $\overline{\mathbf{OM}} = n\overline{\mathbf{OI}}, n \in \mathbb{Q}$.

\etw

\begin{proof} For the proof of the above, the case when $n \in \mathbb{N}$ has been outlined. Consider $n = \frac{p}{q}$, with $p,q \in \mathbb{N}$ (these do not have to be relatively prime). We get $q\overline{\mathbf{OM}} = p\overline{\mathbf{OI}}$, that is, in our parlance, $\overline{\mathbf{OQ}} \cdot \overline{\mathbf{OM}} = \overline{\mathbf{OP}} \cdot \overline{\mathbf{OI}}$. $\overline{\mathbf{OP}}$ and $\overline{\mathbf{OQ}}$ are clearly expressible using the successor formula. The formula capturing the above equality must simply enforce that the same point ($\mathbf{M}$) is the result of both multiplications (and multiplications are expressible). What remains is the case when $n$ is negative. It is enough to make a copy of the triangle forming the coordinate frame's units of measurement (clearly doable) and perform the operations ``in reverse''.\end{proof}

Given any line in a plane, and any coordinate frame, this line can cross both axes at some points; cross one of the axis and be parallel to the other; cross through the origin and be either parallel or cross the line forming the units of measurements at some point; or equal one of the formulas involved in the construction of the coordinate frame. Any intersection points can be captured numerically using the formulas described in the outline of the proof of Theorem \ref{thm:8} (if needs be, changing what counts as \emph{the} axis, see \cite{trybus}) and the remaining parallel cases can also be dealt with as parallelism is expressible in our language.

\btw\label{thm:9}

Let $h,h' \in ROQ(\mathbb{R}^3)$ be half-spaces. Then there is a formula satisfiable in $\mf{M}^3$, such that (1) $h$ satisfies this formula and (2) if $h'$ satisfies the formula, then $h' = h$. 

\etw

\begin{proof} Repeating the construction from Theorem \ref{thm:9} on the other planes of reference means that the resulting compound formula fixes the bounding plane of $h$, yielding $h' = h$ or $h' = -h$. The final disambiguation can be done by insisting that there is no half-space that is contained in $h$ and not contained in $h'$.\end{proof}

Such formulas are sometimes called \emph{fixing} formulas. Note that this notion can be extended from half-planes to arbitrary regions from the domain. In addition we obtain an analogue of the result from \cite{Pratt:1999}. Let us say that a formula is \emph{affine-complete} (in $\mf{M}^3$), if for any two regions $r, r' \in ROQ(\mathbb{R}^3)$ satisfying it, there is an affine transformation mapping $r$ to $r'$ (and, of course, vice versa). Our final result is that --- just as in the case of $\mf{M}^2$ --- every region satisfies an affine-complete formula.   

\btw

Every $r \in ROQ(\mathbb{R}^3)$ satisfies an affine-complete formula in $\mf{M}^3$.

\etw

\begin{proof}Recall that every element of $ROQ(\mathbb{R}^3)$ can be represented as a Boolean combination of half-spaces. Consider a formula stating that certain regions form a coordinate frame (as done above), fixing all the half-spaces that are involved in the construction of $r$ with respect to the resulting coordinate frame (again, as outlined above), and describing the exact Boolean combination resulting in $r$ (clearly expressible). Such a formula has $m+n$ free variables, where $m$ is the number of variables involved in the construction of the coordinate frame and the remaining $n$ variables representing the lines fixed with respect to the coordinate frame (for simplicity, allowing for repetitions of variables in both groups). Consider now an two $m+n$-tuples satisfying this formula. We show that the elements from both tuples are affine equivalent.
Recall that all tetrahedra (essentially: corners in our terminology) are affine-equivalent. Therefore, there is a (unique) affine transformation, say $\tau$, taking the $m$-elements of the first tuple to the $m$-elements of the other. Moreover, since the remaining $n$ half-spaces from the second tuple are fixed with respect to the coordinate frame, by Theorem \ref{thm:9}, the $\tau$-transformed $n$ half-spaces from the first tuple must be the same as the $n$ half-spaces from the second tuple. That is, these are also affine equivalent. Therefore, $\tau$ takes all the elements from the first $m+n$-tuple to the second one. The final formula existentially binds all the variables apart from the one representing $r$.\end{proof}

This result can be easily extended from a single region to formulas of arbitrary arity, as in \cite{Pratt:1999}, thus providing an exact match to Theorem \ref{thm:ian} mentioned above.

\section{Open Problems}

Thus, the stage is set for the task of axiomatising the theory of $\mf{M}^3$. This might be no easy feat, considering how much simpler it is to talk about coordinate frames in $\mf{M}^2$ compared to $\mf{M}^3$. If in the due process, some regularities regarding the constructions are observed, this could be the basis for extending the results to other dimensions. Even at this stage we can note, however, that what has been presented in this paper regarding $\mf{M}^3$ can be most likely extended to any dimension beyond $2$.

\bibliographystyle{unsrtnat}
\bibliography{references}  






\end{document}